\documentclass[12pt]{article}
\usepackage{amsmath,amssymb}
\usepackage{amsthm}

\def\BC{\mathcal{B}}

\def\FC{\mathcal{F}}
\def\HC{\mathcal{H}}

\def\C{\mathbf{C}}
\def\E{\mathbf{E}}

\def\P{\mathbf{P}}
\def\R{\mathbf{R}}

\def\1{\mathbf{1}}

\def\tr{\rm{tr}}

\def\al{\alpha}

\def\pa{\partial}
\def\ep{\epsilon}
\def\de{\delta}
\def\ga{\gamma}
\def\ka{\varkappa}

\newtheorem{prop}{Proposition}[section]
\newtheorem{theorem}{Theorem}[section]
\newtheorem{lemma}{Lemma}[section]

\newtheorem{remark}{Remark}

\newcommand{\si}{\sigma}

\newcommand{\Ga}{\Gamma}

\newcommand{\Om}{\Omega}

\begin{document}
\title{The law of large numbers for quantum stochastic filtering and control of many particle systems}

\author{Vassili N. Kolokoltsov\thanks{Department of Statistics, University of Warwick,
 Coventry CV4 7AL UK, associate member of HSE, Moscow,
  Email: v.kolokoltsov@warwick.ac.uk}}
\maketitle
	
\begin{abstract}
There is an extensive literature on the dynamic law of large numbers for systems of quantum particles,
that is, on the derivation of an equation describing the limiting individual behavior of particles
inside a large ensemble of identical interacting particles. The resulting equations are generally
referred to as nonlinear Scr\"odinger equations or Hartree equations, or Gross-Pitaevski equations.
In this paper we extend some of these
convergence results to a stochastic framework. Concretely we work with the Belavkin stochastic filtering
of many particle quantum systems. The resulting limiting equation is an equation of a new type,
which can be seen as a complex-valued infinite dimensional nonlinear diffusion of McKean-Vlasov type.
This result is the key ingredient for the theory of quantum mean-field games developed by the author
in a previous paper.
\end{abstract}



\section{Introduction}

There is an extensive literature on the dynamic law of large numbers for systems of quantum particles,
that is, on the derivation of an equation describing the limiting individual behavior of particles
inside a large ensemble of identical interacting particles. The resulting equations are generally
referred to as nonlinear Scr\"odinger equations or Hartree equations, or Gross-Pitaevski equations.
The first result of this kind appeared in \cite{Spohn}. Various ingenuous approaches were employed
afterwards in order to extend the conditions of applicability of the convergence results (say, for
various classes of unbounded interacting potentials), exploit various scaling regimes and establishing
the proper rates of convergence, see reviews in \cite{Schlein} and \cite{GolsePaul}.
In this paper, following \cite{KolQuantMFG}, we perform for the first time the extension of some of these
convergence results to a stochastic framework. Concretely we work with the Belavkin stochastic filtering
of many particle quantum systems. The resulting limiting equation is an equation of a new type,
which can be seen as a complex-valued infinite dimensional nonlinear diffusion on a manifold of McKean-Vlasov type.
An optimal control of these equations is beyond the scope of this paper. It leads to slightly more general
HJB  equations on manifolds, as those analysed in \cite{KolDynQuGames}.
 
In \cite{KolQuantMFG} the author merged two recently developed branches of game theory, quantum games and
mean field games (MFGs), creating quantum MFGs. MFGs represent a very popular recent development in game
theory. It was initiated in \cite{HCM3} and \cite{LL2006}. For recent developments one can consult monographs
\cite{BenFr},  \cite{CarDelbook18}, \cite{Gomesbook}, \cite{KolMalbook19} and numerous references therein.
The main ingredient in the rigorous treatment of quantum
MFGs was the quantum stochastic law of large numbers mentioned above stating that the continuously
observed (via nondemolition measurement and quantum filtering) and controlled quantum system of many particles
can be described, in the limit of large number of particles, by a decoupled system of identical nonlinear
stochastic evolutions of individual particles evolved according to some new kind of nonlinear stochastic
Schr\"odinger equation or nonlinear quantum filtering equation. This fact was proved however only under a very
restrictive assumption of the so-called conservative homodyne measurement. In the present paper we obtain
this result in its full generality, namely for arbitrary coupling operators $L$ (not necessarily
anti-Hermitian as in \cite{KolQuantMFG}) with the homodyne detection optical devise. We also obtain
a similar result for jump type stochastic Schr\"odinger equations describing the continuous observations
and filtering of counting type.

 These results make very clear the difference between
  classical and quantum mean-field control. In the former the average used is the measure obtained by averaging
  the Dirac  $\de$-measures arising from the initial conditions of individual particles, and in the latter
the average used is the average of the initial conditions themselves.

The generality chosen implies the necessity to work with the nonlinear quantum filtering equation
leading to a more general limiting equation as in  \cite{KolQuantMFG}, namely to the equation,
which is mostly naturally written in terms of the density matrix $\ga$ as
\[
d\ga_{t}=-i[H+u(t,\ga_{t})\hat H,\ga_{t}] \, dt-i[A^{\bar \eta_t}, \ga_{t}] dt
+(L\ga_{t} L^* -\frac12 L^*L \ga_{t} -\frac12 \ga_{t} L^*L)\, dt
\]
\begin{equation}
\label{eqmainnonlinpartBel1dens0}
+(\ga_{t} L^*+L \ga_{t}-\ga_{t} \, {\tr} (\ga_{t}(L+L^*))) dB_t,
\quad \eta_t(y,z)=\E \ga_{t}(y,z),
\end{equation}
where $H, \hat H,L$ are linear operators in some Hilbert space  $L^2(X)$,  $H, \hat H$ being self-adjoint,
  $A^{\bar \eta_t}$ is the integral operator in $L^2(X)$ with the integral kernel
\[
A^{\bar \eta_t}(x;y)=\int_{X^2} A(x,y;x',y') \overline{\eta_t (y,y')} \, dydy'.
\]
This equation can be looked at as infinite-dimensional complex-valued nonlinear diffusion of McKean-Vlasov type
on the manifold of quantum mechanical states defining a nonlinear Markov process on this manifold in the sense of 
\cite{Kolbook10}.   

In terms of the state vectors $\psi \in L^2(X)$ the limiting equation looks more messy and writes down as
\[
d\psi_t(x) =-i[H +u(t,\psi_t)\hat H +A^{\E (\bar \psi_t \otimes \psi_t)}
-\langle Re \, L\rangle_{\psi} \, Im \, L]  \psi_t(x) \, dt
\]
\begin{equation}
\label{eqmainnonlinco0}
+\frac12 (L-\langle Re \, L\rangle_{\psi_t})^* (L-\langle Re \, L\rangle_{\psi_t}) \psi_t(x) \,dt
+(L-\langle Re \, L \rangle_{\psi_t})\psi_t \, dB_t,
\end{equation}
where, for an operator $L$,
\begin{equation}
\label{eqdefmeanvalueoper}
Re \, L=(L+L^*)/2, \quad Im \, L=(L-L^*)/2i, \quad
\langle L \rangle_v=(v,Lv)/(v,v).
\end{equation}

For the main particular case of a self-adjoint operator $L$  this equation simplifies to the equation
  \[
d\psi_t(x) =-i[H +u(t,\psi_t)\hat H +A^{\E (\bar \psi_t \otimes \psi_t)}]  \psi_t(x) \, dt
\]
\begin{equation}
\label{eqmainnonlinco00}
+\frac12 (L-\langle L\rangle_{\psi_t})^2  \psi_t(x) \,dt
+(L-\langle L \rangle_{\psi_t})\psi_t \, dB_t.
\end{equation}

These equations are different from the nonlinear Schr\"odinger equations that are usually discussed
in the current literature, see e.g.
 \cite{BarbRock16}, \cite{BarbRock18}, \cite{Brzes14}, \cite{GreckGenerNonlinSchr}. In our equation
the nonlinearity depends on the expectations of the correlations calculated with respect to the solution
and thus bearing analogy with the classical McKean-Vlasov nonlinear diffusions.


The content of the paper is as follows. In the next section we recall the basic theory of quantum
continuous measurement and filtering. In the main Section \ref{secmynonlinSchrod} our new nonlinear equations are
introduced for the case of homodyne (or diffusive) detection and the main result on the convergence
of $N$-particle observed quantum evolutions
to the decoupled system of these equations is obtained, together with explicit rates of convergence.
The convergence of solutions is proved here under
the assumption that solutions exist. To repair this drawback, in Section \ref{secwelpos} we prove the well-posedness
of our nonlinear Schr\"odinger equation \eqref{eqmainnonlinco0} including continuous dependence on the
initial conditions with explicit estimates for errors. Since a direct check with Ito's rule shows that
for $\psi$ satisfying   \eqref{eqmainnonlinco0}, the density matrix $\ga_t=\psi_t\otimes \psi_t$ satisfies
 \eqref{eqmainnonlinpartBel1dens0}, it implies the existence of decomposable solutions to
 \eqref{eqmainnonlinpartBel1dens0} which are used in Section   \ref{secmynonlinSchrod}.
In Section \ref{seccounting} the case of a counting observation is analysed, the limiting
stochastic equation of jump-type is introduced and the corresponding convergence result proved.
However, in this case we have to reduce the generality to the case of unitary coupling operators,
because of incompatible intensities of jumps, as explained there.

Appendix A contains a derivation of a quite remarkable technical estimate  that plays the central role
in the proof of our main result.
Appendix B presents certain results on SDEs in Hilbert spaces including McKean-Vlasov type diffusions,
given in the form that is most handy for the application to the well-posedness result of Section \ref{secwelpos}.

\section{Nondemolition observation and quantum filtering}
\label{secnondenmes}

The general theory of quantum non-demolition observation, filtering and resulting
feedback control was built essentially in papers  \cite{Bel87}, \cite{Bel88}, \cite{Bel92}.
For alternative simplified derivations of the main filtering equations given below
(by-passing the heavy theory of quantum filtering) we refer to \cite{BelKol},
\cite{Pellegrini}, \cite{BarchBel}, \cite{Holevo91} and references therein.
For the technical side of organising feedback quantum control in real time,
see e.g. \cite{Armen02Adaptive}, \cite{Bushev06Adaptive} and \cite{WiMilburnBook}.

 We shall describe briefly the main result of this theory.

The non-demolition measurement of quantum systems can be organised in two versions:
photon counting and homodyne detection.

Let us start with the homodyne
(mathematically speaking, diffusive type) detection. Under this type of measurement the
output process $Y_t$ is a usual Brownian motion (under appropriate probability distribution).
The quantum filtering equation can be written in two equivalent ways,
(i) as the linear equation for a non-normalized state:
\begin{equation}
\label{eqqufiBlin}
d\chi_t =-[iH\chi_t +\frac12 L^*L \chi_t ]\,dt+L\chi_t dY_t,
\end{equation}
where the unknown vector $\chi_t$ is from the Hilbert space of the observed quantum system,
which we shall sometimes referred to as the atom,
the self-adjoint operator $H$ is the Hamiltonian
of the corresponding initial (non-observed) quantum evolution and the operator $L$
is the coupling operator of the atom with the optical measurement device specifying the chosen version
of the homodyne detection, or (ii) as the nonlinear equation for the normalized state $\phi=\chi/\|\chi\|$:
\begin{equation}
\label{eqqufiBnonlin}
d\phi_t=-[i(H-\langle Re \, L\rangle_{\phi_t} \, Im \, L)
+\frac12(L-\langle Re \, L\rangle_{\phi_t})^*(L-\langle Re \, L\rangle_{\phi_t})]\phi_t \, dt
+ (L-\langle Re \, L\rangle_{\phi})\phi_t \, dB_t,
\end{equation}
where
\begin{equation}
\label{eqdefinnov}
dB_t=dY_t-\langle L+L^* \rangle_{\chi_t} \, dt=dY_t-\langle L+L^* \rangle_{\phi_t} \, dt
\end{equation}
defines the so-called innovation process $B_t$. It is a direct check via Ito's rule
that \eqref{eqqufiBlin} implies \eqref{eqqufiBnonlin} for $\phi=\chi/\|\chi\|$.

In the most important case of a self-adjoint $L$ these equations simplify to the equations
\begin{equation}
\label{eqqufiBlins}
d\chi_t =-[iH\chi_t +\frac12 L^2 \chi_t ]\,dt+L\chi_t dY_t,
\end{equation}
and, respectively,
\begin{equation}
\label{eqqufiBnonlins}
d\phi_t=-[iH+\frac12(L-\langle L\rangle_{\phi_t})^2]\phi_t \, dt
+ (L-\langle L\rangle_{\phi_t})\phi_t \, dB_t.
\end{equation}

The innovation process is the most natural driving noise to deal with,
because it turns out to be the standard Brownian motion (or the Wiener process)
with respect to the fixed (initial vacuum) state of the homodyne detector.
This fact is one of the key statement of the filtering theory.

In  \cite{KolQuantMFG} we worked with the special case of anti-Hermitian operators $L$, that is $L^*=-L$,
for which nonlinear filtering equation \eqref{eqqufiBnonlins} coincides with the linear one,
the output process coincides with the innovation process, and the linear equation
preserves the norms. Hence this case is referred to as the conservative one, and it is considered as
the less relevant for concrete observations (see discussion in \cite{BarchBook}).

The theory extends naturally to the case of several, say $N$, coupling operators $\{L_j\}$.
For example, the nonlinear quantum filtering equation writes down as
 \[
d\phi_t=-iH \phi_t \, dt + \sum_j [i \langle Re \, L_j\rangle_{\phi_t} \, Im \, L_j
-\frac12(L_j-\langle Re \, L_j\rangle_{\phi_t})^*(L_j-\langle Re \, L_j\rangle_{\phi_t})]\phi_t \, dt
\]
\begin{equation}
\label{eqqufiBnonlinm}
+ \sum_j(L_j-\langle Re \, L_j\rangle_{\phi_t})\phi_t \, dB^j_t,
\end{equation}
where the $N$-dimensional innovation process $B_t=\{B^j_t\}$ is
the standard $N$-dimensional Wiener process connected with the
output process $Y_t=\{Y^j_t\}$ by the equations
\[
dB^j_t=dY^j_t-\langle L_j+L_j^* \rangle_{\chi_t} \, dt.
\]

Recall that the density matrix or density operator $\ga$ corresponding to a unit vector
$\chi\in L^2(X)$ is defined as the orthogonal projection operator on $\chi$.
This operator is usually expressed either as the tensor product  $\ga=\chi\otimes \bar \chi$
or in the most common for physics bra-ket Dirac's notation as $\ga=|\chi\rangle \langle \chi|$.
Of course in the tensor notation $\ga$ is formally an element of the tensor product $L^2(X^2)$.
However, considered as an integral kernel, it is identified with the corresponding integral operator.

As one checks by direct application of Ito's formula,
in terms of the density matrix $\ga$, equation
\eqref{eqqufiBnonlin} rewrites as
\begin{equation}
\label{eqqufiBnonlindens}
d\ga_t=-i[H,\ga_t]\, dt
+(L\ga_t L^* -\frac12 L^*L \ga_t -\frac12 \ga_t L^*L)\, dt
+[\ga_tL^*+L\ga_t-\ga_t \, {\tr} (\ga_t(L^*+L))]\, dB_t,
\end{equation}
and this is the equation we shall work with proving the convergence result.

Let us make here some comments about the evolution of traces.
Equation \eqref{eqqufiBnonlindens} is best suited for dealing with matrices $\ga$ with unit trace.
A convenient extension to arbitrary traces can be written as
\begin{equation}
\label{eqqufiBnonlindens6}
d\ga_t=-i[H,\ga_t]\, dt
+(L\ga_t L^* -\frac12 L^*L \ga_t -\frac12 \ga_t L^*L)\, dt
+[\ga_tL^*+L\ga_t-\frac{\ga_t}{{\tr} \, \ga_t} \, {\tr} (\ga_t(L^*+L))]\, dB_t.
\end{equation}
It is seen directly from this equation that $d\, {\tr} \, \ga_t=0$ and hence the trace is preserved
by evolution \eqref{eqqufiBnonlindens6} (as long as it is well-posed of course).
For equation \eqref{eqqufiBnonlindens} one gets
\begin{equation}
\label{eqqufiBnonlindenstr}
d\, {\tr} \, \ga_t= {\tr} (\ga_t(L^*+L))(1-  {\tr} \, \ga_t) \, dB_t.
\end{equation}
Thus  evolution \eqref{eqqufiBnonlindens} does not preserve traces in general. But
${\tr} \, \ga_t=1$ is a solution of \eqref{eqqufiBnonlindenstr}. Hence, if the initial
condition has  ${\tr} \, \ga_0=1$ and the equation \eqref{eqqufiBnonlindenstr} is well posed, then
the solutions to  \eqref{eqqufiBnonlindens} do preserve the trace.

Similar remarks on the preservation of norms of equation \eqref{eqqufiBnonlin} are worth mentioning.
Calculating $d(\phi,\phi)$ from this equation we see first that the terms with $H$ and $Im \, L$ disappear
and then, by Ito's rule, all terms with differential $dt$ vanish. Therefore
\[
d(\phi_t,\phi_t)=((L-\langle Re \, L\rangle_{\phi_t})\phi_t,\phi_t) \, dB_t
+(\phi_t,(L-\langle Re \, L\rangle_{\phi_t})\phi_t)\, dB_t
\]
\begin{equation}
\label{eqBelpreservenorm}
=((L+L^*-2\langle Re \, L\rangle_{\phi_t})\phi_t, \phi_t)  \, dB_t
=2((Re \, L-\langle Re \, L\rangle_{\phi_t})\phi_t, \phi_t)  \, dB_t=0,
\end{equation}
since $\langle Re \, L\rangle_{\phi_t}=(\phi_t, (Re \, L) \phi_t)/(\phi_t,\phi_t)$.
Hence solutions to \eqref{eqqufiBnonlin} preserve the norm almost surely.

If instead of $\langle Re \, L\rangle_{\phi}$ in  equation \eqref{eqqufiBnonlin} one would use
the non-normalized expression $(\phi, (Re \, L) \phi)$ (as some authors do), one would get
\[
d(\phi_t,\phi_t)=(\phi_t, (Re \, L) \phi_t)(1-(\phi_t, \phi_t)) dB_t,
\]
so that for the initial condition with $\|\phi_0\|=1$ one would conclude that $\|\phi_t\|=1$ by uniqueness, as
for the density matrix case.

Let us briefly describe the case of counting measurements.
In this case the main equation of quantum filtering takes the form
\begin{equation}
\label{eqBeleqcountm}
d\ga_t=- i[H, \ga_t]\, dt +\sum_j (-\frac12 \{L^*_jL_j,\ga_t\}+ {\tr} (L_j\ga_t L^*_j) \ga_t ) \, dt
+\sum_j\left(\frac{L_j\ga_t L^*_j}{{\tr} (L_j\ga_t L^*_j)}-\ga_t\right) dN^j_t,
\end{equation}
where $H$ is again the Hamiltonian of the free (not observed) motion
of the quantum system, the operators $\{L_j\}$ define the coupling of the system
with the measurement devices,  and  the counting (observed) processes $N_t^j$ are
independent and have the position dependent intensities ${\tr} (L^*_jL_j\ga_t)$,
so that the  compensated processes $M_t^j=N_t^j-\int_0^t {\tr} (L_j^*L_j\ga_s) \, ds$ are martingales.
In terms of the compensated processes $M_t^j$ equation \eqref{eqBeleqcountm} rewrites as
\begin{equation}
\label{eqBeleqcountm1}
d\ga_t=- i[H, \ga_t]\, dt +\sum_j (L_j\ga L_j^*-\frac12 \{L^*_jL_j,\ga_t\} ) \, dt
+\sum_j\left(\frac{L_j\ga_t L^*_j}{{\tr} (L_j\ga_t L^*_j)}-\ga_t\right) dM^j_t.
\end{equation}

As in the case of diffusive measurements, this dynamics preserves the set of pure states.
Namely, if $\phi$ satisfies the equation
\begin{equation}
\label{eq2thquantumfileqcountpure}
d \phi_t=-\left(iH +\frac12 \sum_j((L_j^*-1)L_j-(L_j-\|L_j\phi_t\|^2))\right)\phi_t \, dt
+\sum_j \left(\frac{L_j\phi_t}{\|L_j\phi_t\|}-\phi_t\right)  dM_t^j,
\end{equation}
then $\ga_t=\phi_t\otimes \bar \phi_t$ satisfies equation \eqref{eqBeleqcountm1}.

\begin{remark} Equation \eqref{eqBeleqcountm} is slightly nonstandard as the driving noises $N_t^j$ are
position dependent. However there is a natural way to rewrite it in terms of independent driving noises.
Namely, with a standard independent Poisson random measure processes $N^j(dx \,dt)$ on $\R_+\times \R_+$
(with Lebesgue measure as intensity) one can rewrite equation \eqref{eqBeleqcountm} in the following equivalent form:
\[
d\ga_t=- i[H, \ga_t]\, dt +\sum_j ( -\frac12 \{L_j^*L_j,\ga_t\}+ {\tr} (L_j\ga_t L_j^*) \ga_t ) dt
\]
\[
+\sum_j \left(\frac{L_j\ga_t L_j^*}{{\tr} (L_j\ga_t L_j^*)}-\ga_t\right) \1({\tr} (L_j^*L_j\ga_t)\le x) N^j(dx\, dt),
\]
see details of this construction in \cite{Pellegrini10a}. Alternatively, one can make sense
of \eqref{eqBeleqcountm} in terms of the general theory of weak SDEs from \cite{Kol11}.
\end{remark}

The theory of quantum filtering reduces the analysis of quantum dynamic control and games
to the controlled version of evolutions \eqref{eqqufiBnonlindens}. The simplest situation concerns the case
 when the homodyne device is fixed, that is the operators $L_j$ are fixed, and the players can control the
 Hamiltonian $H$, say, by applying appropriate electric or magnetic fields to the atom. Thus equations
 \eqref{eqqufiBnonlindens} or \eqref{eqBeleqcountm} become modified by allowing $H$ to depend on one
 or several control parameters.

\section{The limiting equation: diffusive measurement}
\label{secmynonlinSchrod}

Let $X$ be a Borel space with a fixed Borel measure that we denote $dx$.
For a linear operator $O$ in $L^2(X)$ we shall denote by $O_j$ the operator
in $L^2(X^N)$ that acts on functions $f(x_1, \cdots, x_N)$ as $O$ acting on the variable $x_j$.

Let $H$ be a self-adjoint operator in $L^2(X)$
and $A$ a self-adjoint integral operator in $L^2(X^2)$ with the kernel $A(x,y;x',y')$ that acts
on the functions of two variables as
\[
A\psi(x,y)=\int_{X^2}   A(x,y;x',y') \psi(x',y') \, dx' dy'.
\]
It is assumed that $A$ is symmetric in the sense that it takes
symmetric functions $\psi(x,y)$ (symmetric with respect to permutation of $x$ and $y$) to
symmetric functions. In terms of the kernel this means the equation
\[
A(x,y;x',y')=A(y,x;y',x').
\]

Let us consider the quantum evolution of $N$ particles
driven by the standard interaction Hamiltonian
\begin{equation}
\label{eqHambinaryinter0}
 H(N)f(x_1, \cdots , x_N)=\sum_{j=1}^N H_jf(x_1, \cdots , x_N)
 + \frac{1}{N}\sum_{i<j\le N} A_{ij}f(x_1, \cdots , x_N),
\end{equation}
with $A_{ij}$ denoting the action of $A$ on the variables $x_i,x_j$.

\begin{remark} Alternatively, one can consider a multiplication operator,
by a symmetric function $V(x_i,x_j)$, rather than the integral operator $A$  in \eqref{eqHambinaryinter0}.
\end{remark}

Assume further that this quantum system is observed
via coupling with the collection of identical one-particle operators $L$.
 That is, we consider the filtering equation \eqref{eqqufiBnonlinm} of the type
\[
d\Psi_{N,t} =
  \sum_j [i \langle Re \, L_j\rangle_{\phi} \, Im \, L_j
 -\frac12 (L_j-\langle Re \, L_j\rangle_{\Psi_{N,t}})^*(L_j-\langle Re \, L_j\rangle_{\Psi_{N,t}})]\Psi_{N,t} \, dt
\]
\begin{equation}
\label{eqmainNpartBelnonls}
-iH(N) \Psi_{N,t} \, dt+\sum_{j=1}^N (L_j-\langle Re \, L_j\rangle_{\Psi_{N,t}})\Psi_{N,t} \, dB^j_t.
\end{equation}

We shall show that as $N\to \infty$ the solution of this equation
and the corresponding density matrices $\Ga_{N,t}=\Psi_{N,t}\otimes \overline{\Psi_{N,t}}$
are close to the product
of the solutions of certain one-particle nonlinear stochastic equations, more precisely, that
the partial traces $\Ga^{(j)}_{N,t}$ of $\Ga_{N,t}$ are close to the density matrices $\ga_{j,t}$
of the individual equations.

However, for the study of control we need a more general setting. We shall assume that
the individual Hamiltonian $H$ has a control component, that is, it can be written as $H+u\hat H$ with two self-adjoint
operators $H$ and $\hat H$ and $u$ a real control parameter taken from a bounded interval $[-U,U]$.
 Suppose that, for the idealized limiting evolution, $u$ is chosen as a
 certain function of an observed density matrix $\ga_{j,t}$, that is $u=u(t,\ga_{j,t})$.
Then in the original $N$ particle evolution $u$ will be chosen based on the approximation $\Ga^{(j)}_{N,t}$ to $\ga_{j,t}$,
that is as $u=u(t,\Ga^{(j)}_{N,t})$. Thus the controlled and observed $N$-particle evolution will be given by equation
 \eqref{eqmainNpartBelnonls} with the nonlinear controlled Hamiltonian $H_u(N)$ instead of $H$:

\begin{equation}
\label{eqHambinaryinter1}
 H_u(N)f(x_1, \cdots , x_N)=\sum_{j=1}^N (H_j+u(t,\Ga^{(j)}_N) \hat H_j)f(x_1, \cdots , x_N)
 + \frac{1}{N}\sum_{i<j\le N} A_{ij}f(x_1, \cdots , x_N),
\end{equation}
where $u(t,\ga)$ is some continuous function.

The corresponding density matrix $\Ga_{N,t}=\Psi_{N,t}\otimes \overline{\Psi_{N,t}}$ satisfies the
equation of type \eqref{eqqufiBnonlindens}:
\[
d \Ga_{N,t}
=-i [H_u(N),\Ga_{N,t}] dt-\frac{i}{N} \sum_{l<j\le N} [A_{lj},\Ga_{N,t}] \, dt
+\sum_j (L_j\Ga_{N,t} L_j^* -\frac12 L^*_j L_j \Ga_{N,t} -\frac12 \Ga_{N,t} L^*_j L_j)\, dt
\]
\begin{equation}
\label{eqmainNpartBeldensnonl}
+\sum_j (\Ga_{N,t} L_j^*+L_j \Ga_{N,t}-\Ga_{N,t} \, {\tr} (\Ga_{N,t}(L_j^*+L_j))) dB_t^j.
\end{equation}

Our main objective in this paper is to prove that, as $N\to \infty$, the limiting evolution of each particle is described
by the nonlinear stochastic equation
\[
d\psi_{j,t}(x) =-i[H +u(t,\ga_{j,t})\hat H +A^{\bar \eta_t}
-\langle Re \, L\rangle_{\phi} \, Im \, L]  \psi_{j,t}(x) \, dt
\]
\begin{equation}
\label{eqmainnonlinco}
+\frac12 (L-\langle Re \, L\rangle_{\psi_{j,t}})^* (L-\langle Re \, L\rangle_{\psi_{j,t}}) \psi_{j,t}(x)\,dt
+(L-\langle Re \, L \rangle_{\psi_{j,t}})\psi_{j,t} \, dB^j_t,
\end{equation}
where $A^{\bar \eta_t}$ is the integral operator in $L^2(X)$ with the integral kernel
\[
A^{\bar \eta_t}(x;y)=\int_{X^2} A(x,y;x',y') \overline{\eta_t (y,y')} \, dydy'
\]
and
\[
\eta_t (y,z)=\E (\psi_{j,t} (y) \bar \psi_{j,t}(z)).
\]
The equation for the corresponding density matrix $\ga_{j,t}=\psi_{j,t} \otimes \bar \psi_{j,t}$ writes down as
\[
d\ga_{j,t}=-i[H+u(t,\ga_{j,t})\hat H,\ga_{j,t}] \, dt-i[A^{\bar \eta_t}, \ga_{j,t}] dt
+(L\ga_{j,t} L^* -\frac12 L^*L \ga_{j,t} -\frac12 \ga_{j,t} L^*L)\, dt
\]
\begin{equation}
\label{eqmainnonlinpartBel1dens}
+(\ga_{j,t} L^*+L \ga_{j,t}-\ga_{j,t} \, {\tr} (\ga_{j,t}(L+L^*))) dB^j_t,
\quad \eta_t(y,z)=\E (\psi_{j,t} (y) \bar \psi_{j,t}(z))=\E \ga_{j,t}(y,z).
\end{equation}

The result extends directly to the case when the measurement related to each particle is multi-dimensional,
that is the operator $L$ is vector-valued, $L=(L^1, \cdots, L^k)$, in which case each noise $dB^j_t$ is also
$k$-dimensional, so that the last term in \eqref{eqmainnonlinco}  should be understood as the inner product:
\[
(L-\langle L\rangle_{\psi_{j,t}})\psi_{j,t} \, dB^j_t
=\sum_{l=1}^k(L^l-\langle L^l\rangle_{\psi_{j,t}})\psi_{j,t} \, dB^j_{l,t},
\]
with all other terms containing $L$ understood in the same way.

Our analysis will be carried out via the extension of the method suggested
by Pickl in a deterministic case,
see \cite{Pickl} and  \cite{KnowlesPickl}, to the present stochastic framework.
In Pickl's approach the main measures of the deviation of the solutions $\Psi_{N,t}$ to $N$-particle
systems from the product of the solutions $\psi_t$ to the Hartree equations are the following positive numbers
from the interval $[0,1]$:
\[
\al_N(t)=1-(\psi_t, \Ga_{N,t} \psi_t).
\]

In the present stochastic case, these quantities depend not just on the number of particles
in the product, but on the concrete choice of these particles. The proper stochastic analog
of the quantity $\al_N(t)$ is the collection of random variables
\begin{equation}
\label{eqforalpha}
\al_{N,j}(t)=  1-(\psi_{j,t}, \Ga_{N,t} \psi_{j,t})=1-{\tr}(\ga_{j,t} \Ga_{N,t})=1-{\tr}(\ga_{j,t} \Ga^{(j)}_{N,t}),
\end{equation}
where the latter equation holds by the definition of the partial trace.
Here $\ga_{j,t}$ is identified with the operator in $L^2(X^N)$ acting on the $j$th variable and
$\Ga^{(j)}_{N,t}$ denotes the partial trace of $\Ga_{N,t}$ with respect to all variables except for the $j$th.

The key property of equations \eqref{eqmainNpartBeldensnonl} and \eqref{eqmainnonlinpartBel1dens} is that
their solutions preserve the set of operators with the unit trace
(see discussion around equation \eqref{eqqufiBnonlindenstr} above). Hence \eqref{eqforalpha}
rewrites as

\begin{equation}
\label{eqforalpha1}
\al_{N,j}(t)={\tr}((\1-\ga_{j,t}) \Ga_{N,t})={\tr}((\1-\ga_{j,t}) \Ga^{(j)}_{N,t}).
\end{equation}

Due to the i.i.d. property of the solutions to \eqref{eqmainnonlinco},
the expectations $\E \al_N(t)=\E \al_{N,j}(t)$ are well defined (they do not
depend on a particular choice of particles).

Expressions $\al_{N,j}$ can be linked with the traces by the following inequalities, due to Knowles and Pickl:
\begin{equation}
\label{ineqKnPi}
\al_{N,j}(t)\le {\tr} |\Ga^{(j)}_{N,t}-\ga_{j,t}|\le 2\sqrt{2\al_{N,j}(t)},
\end{equation}
see Lemma 2.3 from \cite{KnowlesPickl}.

\begin{theorem}
\label{thmynonlinSch}
Let $H,\hat H,L$ be operators in $L^2(X)$, with $H,\hat H$ self-adjoint, and $\hat H, L$ bounded.
Let $A$ be a symmetric self-adjoint integral operator $A$ in $L^2(X^2)$ with a Hilbert-Schmidt
kernel, that is a kernel $A(x,y;x',y')$ such that
\begin{equation}
\label{eq1thmynonlinSch}
\|A\|^2_{HS}=\int_{X^4} |A(x,y;x',y')|^2 \, dx dy dx'dy' <\infty,
\end{equation}
\begin{equation}
\label{eq1athmynonlinSch}
A(x,y;x'y')=A(y,x;y',x'), \quad A(x,y;x',y')=\overline{A(x',y';x,y)}.
\end{equation}

Let the function $u(t,\ga)$ with values in a bounded interval $[-U,U]$
be Lipschitz in the sense that
\begin{equation}
\label{eq1thmynonlinSchco}
|u(t,\ga)-u(t,\tilde \ga)|\le \ka \, {\tr} |\ga -\tilde \ga|.
\end{equation}

 Let $\psi_{j,t}$ be solutions to equations
\eqref{eqmainnonlinco} with i.i.d. initial conditions
$\psi_{j,0}$, $\|\psi_{j,0}\|=1$.
Let $\Psi_{N,t}$ be a solution to the $N$-particle equation \eqref{eqmainNpartBelnonls}
with $H(N)$ of type \eqref{eqHambinaryinter1}, with some initial condition
$\Psi_{N,0}$, $\|\Psi_{N,0}\|_2=1$ such that
\[
\al_N(0)= \al_{N,j}(0)=1-\E \, {\tr} (\ga_{j,0} \Ga_{N,0})=1-\E \, {\tr} (\ga_{j,0} \Ga^{(j)}_{N,0})
\]
are equal for all $j$. The main example of such initial condition is of course the product
\[
\Psi_{N,0}=\prod \psi_{j,0}(x_j),
\]
where $\al_{N,j}(0)=0$ for all $j$.
Then
\[
\E \al_N(t) \le \exp\{ (7\|A\|_{HS}+6 \ka \|\hat H\|+28 \|L\|^2)t\} \E \al_{N}(0)
\]
\begin{equation}
\label{eq2thmynonlinSchco}
+(\exp\{(7\|A\|_{HS}+6 \ka \|\hat H\|+28 \|L\|^2)t\}-1)\frac{1}{\sqrt N}.
\end{equation}

\end{theorem}

By \eqref{eqforalpha1} it follows that if $\al_{N}(0) \to 0$, as $N\to \infty$ (for instance if $\al_{N}(0)=0$),
then $\E \, {\tr} |\Ga^{(j)}_{N,t}\to \ga_{j,t}| \to 0$, as $N\to \infty$.
Of course, everything remains unchanged for a vector-valued $L$.

\begin{proof}

Using definition \eqref{eqforalpha}  and Ito's product rule we derive that
\begin{equation}
\label{eqforalphaconder}
d \al_{N,j}(t)=-{\tr} (d\Ga_{N,t} \ga_{j,t})-{\tr} (\Ga_{N,t} \, d \ga_{j,t})-{\tr} (d\Ga_{N,t} \, d \ga_{j,t})
=(C_j+D_j)dt+\sum_j F_j dB^j_t,
\end{equation}
where we denoted by $C_j$ the part of the differential, which does not depend on $L_j$, and by $D_j$ the part
that depends on $L_j$. These parts clearly separate, so that $D_j$ does not contain $H,\hat H$ and $A_{jk}$.

We are going to estimate $\E |\dot \al_{N,t}|$, and thus we are interested only in the estimates for $|C_j|$ and $|D_j|$.
The part $C_j$ was calculated and estimated in \cite{KolQuantMFG} (proofs of Theorems 3.1 and 3.2) yielding the estimate
\begin{equation}
\label{eqestimfromprebious}
\E |C_j| \le  7 \|A\|_{HS} \left(\E \, \al_N(t)+\frac{1}{\sqrt N}\right)
+4\sqrt 2 \ka \|\hat H\|   \E \al_N(t).
\end{equation}

As was shown in \cite{KolQuantMFG}, $D_j$ vanishes for the case $L_j=-L_j^*$. For the general case we derive from
\eqref{eqforalphaconder},  \eqref{eqmainnonlinpartBel1dens} and \eqref{eqmainNpartBeldensnonl} that
(we omit the index $t$ for brevity below)
\[
D_j={\tr} \bigl[ \sum_k (\frac12 L^*_k L_k\Ga_N +\frac12 \Ga_N L_k^*L_k-L_k \ga_N L_k^*)\ga_j
+\Ga_N (\frac12 L^*_j L_j \ga_j +\frac12 \ga_j L_j^*L_j-L_j \ga_j L_j^*)
\]
\[
-(\Ga_N L_j^*+L_j \Ga_N -\Ga_N \, {\tr} (\Ga_N(L_j^*+L_j)))(\ga_j L_j^*
+L_j \ga_j -\ga_j \, {\tr} (\ga_j(L_j^*+L_j)))\bigr],
\]
where we omitted the index $t$ for brevity.

The terms with $k\neq j$ vanish, because
\[
{\tr} \, (L_k^*L_k \Ga_N \ga_j)={\tr} \, (\Ga_N L_k^*L_k \ga_j)
={\tr} \, (L_k \Ga_N L_k^*\ga_j)={\tr} \, (\Ga_N \ga_jL_k^*L_k).
\]
We are left with
\[
D_j={\tr} \bigl[  \frac12 L^*_j L_j\Ga_N \ga_j +\frac12 \Ga_N L_j^*L_j\ga_j-L_j \Ga_N L_j^*\ga_j
+\frac12 \Ga_N L^*_j L_j \ga_j +\frac12 \Ga_N\ga_j L_j^*L_j-\Ga_NL_j \ga_j L_j^*
\]
\[
-\Ga_N L_j^*\ga_j L_j^*-\Ga_N L_j^*L_j \ga_j-L_j \Ga_N \ga_j L_j^*-L_j \Ga_N L_j\ga_j
\]
\[
+(\Ga_N L_j^*+L_j \Ga_N)\ga_j \, {\tr} (\ga_j(L_j^*+L_j))
+\Ga_N \, {\tr} (\Ga_N(L_j^*+L_j))(\ga_j L_j^*+L_j \ga_j)
\]
\[
-\Ga_N \ga_j\, {\tr} (\Ga_N(L_j^*+L_j))\, {\tr} (\ga_j(L_j^*+L_j))\bigr].
\]
Further cancelation yields the following rather awkwardly looking expression:
\[
D_j=-{\tr} (\ga_j L_j \Ga_N L_j^*+\ga_j L_j^* \Ga_N L_j +\ga_j L_j^*\Ga_N L_j^*+\ga_j L_j \Ga_N L_j)
\]
\[
+ {\tr} (\ga_j\Ga_N L_j^*+\ga_j L_j \Ga_N) \, {\tr} (\ga_j(L_j^*+L_j))
+ {\tr} (\ga_j \Ga_N L_j+ \ga_j L_j^* \Ga_N) \, {\tr} (\Ga_N(L_j^*+L_j))
\]
\[
-{\tr} (\Ga_N \ga_j)\, {\tr} (\Ga_N(L_j^*+L_j))\, {\tr} (\ga_j(L_j^*+L_j)).
\]
However, by Lemma \ref{lemmaonspectraces} from Appendix,
\[
|D_j| \le  28 \|L\|^2 \, {\tr}\, ((1-\ga_j) \Ga_N).
\]
Combining this estimate with \eqref{eqestimfromprebious} it follows that
\[
|\dot \E \al_{N,t} \le  7 \|A\|_{HS} \left(\E \, \al_N(t)+\frac{1}{\sqrt N}\right)
+(4\sqrt 2 \ka \|\hat H\|+28\|L\|^2)   \E \al_N(t).
\]
Applying Gronwall's lemma yields \eqref{eq2thmynonlinSchco}.

\end{proof}

\section{A well-posedness result}
\label{secwelpos}

This section is devoted to the well-posedness of equation \eqref{eqmainnonlinco0}:

\[
d\psi_t(x) =-i[H \psi_t(x)+u(t,\psi_t)\hat H +A^{\E (\bar \psi_t \otimes \psi_t)} \psi_t(x)
-\langle Re \, L\rangle_{\phi} \, Im \, L]  \psi_t(x) \, dt
\]
\begin{equation}
\label{eqmainnonlinco0rep}
+\frac12 (L-\langle Re \, L\rangle_{\psi_t})^* (L-\langle Re \, L\rangle_{\psi_t}) \psi_t(x) \,dt
+(L-\langle Re \, L \rangle_{\psi_t})\psi_t \, dB_t.
\end{equation}

In this analysis we shall use the results of Appendix B and the notations
for spaces and processes introduced there.
Equation \eqref{eqmainnonlinco0rep} is of type \eqref{eqMcKeVlBan},
with the Hilbert space $\HC=L^2(X)$ and $A=-iH$.

As an auxiliary tool we shall study
the more standard equations
\[
d\psi_t(x) =-i[H \psi_t(x)+u(t,\psi_t)\hat H +A^{\xi_t} \psi_t(x)
-\langle Re \, L\rangle_{\phi} \, Im \, L]  \psi_t(x) \, dt
\]
\begin{equation}
\label{eqmainnonlinco0rep1}
+\frac12 (L-\langle Re \, L\rangle_{\psi_t})^* (L-\langle Re \, L\rangle_{\psi_t}) \psi_t(x) \,dt
+(L-\langle Re \, L \rangle_{\psi_t})\psi_t \, dB_t,
\end{equation}
with $\xi_t(y,z)$ a given continuous curve in $\HC\otimes \HC$ and
\[
d\psi_t(x) =-i[H \psi_t(x)+u_t\hat H +A^{\xi_t} \psi_t(x)
-\langle Re \, L\rangle_{\phi} \, Im \, L]  \psi_t(x) \, dt
\]
\begin{equation}
\label{eqmainnonlinco0rep2}
+\frac12 (L-\langle Re \, L\rangle_{\psi_t})^* (L-\langle Re \, L\rangle_{\psi_t}) \psi_t(x) \,dt
+(L-\langle Re \, L \rangle_{\psi_t})\psi_t \, dB_t,
\end{equation}
with $u_t$ a given continuous function.

These equations are of type \eqref{eqMcKeVlBanau} of Appendix B.

The results of Appendix B are not directly applicable to equations
\eqref{eqmainnonlinco0rep} and \eqref{eqmainnonlinco0rep1}, because the
Lipschitz continuity is only local for terms containing $u$ and $\xi$.
However, as was proved in \cite{BarchBook}, the terms containing $L$ are
uniformly Lipschitz for bounded $L$.

\begin{remark}
The proof of this fact is obtained in \cite{BarchBook} by elementary direct estimates, which
were however rather lengthy and not very intuitive. To see this fact more
directly one can just note that the derivatives of all coefficients involving
$L$ in  \eqref{eqmainnonlinco0rep1} are bounded. For instance, the last coefficient
in  \eqref{eqmainnonlinco0rep1} has the form
\[
f(\psi)=\frac{(\psi,M\psi)}{(\psi,\psi)}\psi
\]
with a self-adjoint $M$. Taking a simpler situation with a real Hilbert space we find
\[
\frac{\pa f}{\pa \psi}
 =\frac{(\psi,M\psi)}{(\psi,\psi)} \1
 +2\frac{M\psi\otimes \psi}{(\psi,\psi)}-2\frac{(M\psi, \psi)\psi\otimes \psi}{(\psi,\psi)^2}
\]
so that
\[
\|\frac{\pa f}{\pa \psi} \| \le 5\|M\|.
\]
\end{remark}

\begin{theorem}
\label{thwelpos}
Let $H,\hat H,L$ be self-adjoint operators in $L^2(X)$, with $\hat H$ and $L$ being bounded,
and $A$ be an integral operator in $L^2(X^2)$ with a kernel $A(x,y;x',y')$ satisfying
\eqref{eq1thmynonlinSch} and \eqref{eq1athmynonlinSch}.

Let the function $u(t,\ga)$ with values in a bounded interval $[-U,U]$  be Lipschitz
either as a function of $\psi$ or as a function of $\psi\otimes \bar \psi$,
that is either
\begin{equation}
\label{eq1thwelpos}
|u(t,\psi)-u(t,\tilde \psi)|\le \ka \|\psi -\tilde \psi\|,
\end{equation}
or
\begin{equation}
\label{eq2thwelpos}
|u(t,\psi)-u(t,\tilde \psi)|\le \ka (\|\psi\|+\|\tilde \psi\|)\|\psi -\tilde \psi\|.
\end{equation}
Then (i) equation \eqref{eqmainnonlinco0rep} is globally well-posed in the sense of mild solutions;

(ii) for any continuous curve $\xi_t$ in $(L^2(X))^{\otimes 2}$ equation \eqref{eqmainnonlinco0rep1}
is globally well-posed in the sense of mild solutions;

(iii) for any continuous curve $\xi_t$ in $(L^2(X))^{\otimes 2}$
and a continuous function $u_t$ equation \eqref{eqmainnonlinco0rep2}
is globally well-posed in the sense of mild solutions;

(iv) solutions $\psi_t$ to equation \eqref{eqmainnonlinco0rep1} depend continuously on the initial conditions
in the following precise sense: for two solutions with the initial conditions $Y_1$ and $Y_2$ we have the estimate
\begin{equation}
\label{eq3thwelpos}
\|\psi^1-\psi^2\|_{ad,T}\le  2 e^{aT+bT^2} \|Y_1 -Y_2\|,
\end{equation}
where $a,b$ are constants depending on $\ka, U, \|\hat H\|, \| L\|, \|A\|_{HS}$.

(v) for all these equations to hold in the strong sense it is sufficient to assume that there
  exists an invariant core $D$ of the group $e^{iH}$ and a norm $\|.\|_D$ on it such that $\|.\|_D\ge \|.\|$,
  $D$ is a Banach space under this norm and the operators $\hat H, L, A^{\xi}$ are bounded operators
  $L^2(X)\to D$, for the latter operator uniformly with respect to bounded sets of curves $\xi_t$,
  and the initial condition is taken from $D$.
\end{theorem}

\begin{proof}
Due to the remark on the Lipschitz continuity of the coefficients involving $L$,
equation \eqref{eqmainnonlinco0rep2} satisfies the conditions of Proposition \ref{propSDEBanach},
and hence statement (iii) is a consequence of this Proposition. Moreover, by Proposition \ref{propSDEBanach3},
(v) follows from (i)-(iii). To show (i) and (ii) we note that the solutions to \eqref{eqmainnonlinco0rep2}
preserve norms almost surely. This is a standard fact shown by calculations \eqref{eqBelpreservenorm}.

\begin{remark}
Calculations \eqref{eqBelpreservenorm} were performed assuming the strong (not just mild) form of the equation.
That is, strictly speaking under assumptions of statement (iv) only. However, by approximating operators
$\hat H, L, A^{\xi}$ by operators satisfying assumptions of (iv), one obtains this result for the general case.
\end{remark}

In particular, looking for the solutions of equations \eqref{eqmainnonlinco0rep} and
\eqref{eqmainnonlinco0rep1} we can reduce our attentions to curves $\xi_t$
such that $\|\xi_t\|\le 1$.
Therefore comparing two mild solutions $\psi^1_t$ and $\psi^2_t$ of equation \eqref{eqmainnonlinco0rep2} with the same
initial condition $\psi_0$, $\|\psi_0\|=1$, but different pairs $(\xi^1_t, u^1_t)$, $(\xi^2_t, u^2_t)$
such that $\|\xi^j_t\|\le 1$ we find,
similarly to calculations performed in the proof of Proposition  \ref{propSDEBanach}, that
\[
\E \|\psi^1_t-\psi^2_t\|^2
\le (a +b t)\E \int_0^t \|\psi^1_s - \psi^2_s\|^2\, ds
\]
\[
+ \|\hat H\|^2 (\int_0^t |u^1_s-u^2_s| ds)^2
+\|A\|^2_{HS}  (\int_0^t \|\xi^1_s-\xi^2_s\| ds)^2
\]
\[
\le  (a +b t)\E \int_0^t \|\psi^1_s - \psi^2_s\|^2\, ds
+\|\hat H\|^2 t \|u^1_.-u^2_.\|_t^2
+\|A\|^2_{HS} t\|\xi^1_.-\xi^2_.\|_t^2,
\]
where $a,b$ are constants depending on $U, \|\hat H\|, \| L\|, \|A\|_{HS}$,
and where we denote by $\|.\|_t$ the maximum of the corresponding norms
of the functions at all times $s\in [0,t]$.

Consequently, by Gronwall's lemma and because $\psi^1_0=\psi^2_0$,
\[
\E \|\psi^1_t-\psi^2_t\|^2 \le e^{at+bt^2} t (\|\hat H\|^2 \|u^1_.-u^2_.\|_t^2
+\|A\|^2_{HS} \|\xi^1_.-\xi^2_.\|_t^2)
\]
and thus
\[
\|\psi^1_.-\psi^2_.\|_{ad,t} \le e^{(at+bt^2)/2} \sqrt t (\|\hat H\| \|u^1_.-u^2_.\|_t
+\|A\|_{HS} \|\xi^1_.-\xi^2_.\|_t).
\]
Thus, by choosing $t$ small enough, one can make the Lipschitz constant of the mapping
 $(\xi_., u_.) \to X_.$ arbitrary small.
Since $u(t,\ga)$ is Lipschitz and the mapping $X \to \E (\bar X \otimes X)$ is Lipschitz
(uniformly for $X$ with $\|X\|=1$) we find that the composition mapping
\[
 (\xi_., u_.) \to \psi_. \to   \E (\bar \psi_. \otimes \psi_.)
 \]
 is a contraction for small times. Hence it has a unique fixed point. By iteration
 we build a unique solution to    \eqref{eqmainnonlinco0rep} proving (i).

 Statement (ii) is proved analogously.

It remains to show (iv). For two solutions
$\psi^1_t$ and $\psi^2_t$ of equation \eqref{eqmainnonlinco0rep1} with some unit
initial condition $Y_1$, $Y_2$ and some $\xi_t$ with $\|\xi_t\|\le 1$ we find,
similarly to above, that
\[
\E \|\psi^1_t-\psi^2_t\|^2
\le 2\|Y_1-Y_2\|^2+(a +b t)\E \int_0^t \|\psi^1_s - \psi^2_s\|^2\, ds.
\]
Hence by Gronwall's lemma it follows that
\[
\E \|\psi^1_t-\psi^2_t\|^2 \le 2 e^{at+bt^2} \|Y_1-Y_2\|^2,
\]
implying \eqref{eq3thwelpos}.
\end{proof}

\section{The limiting equation: counting measurement}
\label{seccounting}

The analog of equation \eqref{eqmainNpartBeldensnonl} describing the observation of a collection of
identical quantum particles, arising from the general quantum filtering equation
\eqref{eqBeleqcountm1} or observation of counting type is the equation

\[
d \Ga_{N,t}
=\left(-i\sum_j [H_j,\Ga_{N,t}]-\frac{i}{N} \sum_{l<j\le N} [A_{lj},\Ga_{N,t}]
+\sum_k (L_k\Ga_{N,t} L_k^* -\frac12 L^*_kL_k \Ga_{N,t} -\frac12 \Ga_{N,t} L^*_kL_k)\right)\, dt
\]
\begin{equation}
\label{eqmainNpartBeldensnonlcount}
+ \sum_k\left(\frac{L_k \Ga_{N,t} L^*_k}{{\tr} \, (\Ga_{N,t} L^*_kL_k)}- \Ga_{N,t}\right) \, dM^k_t,
\end{equation}
where $M^l_t$ are martingales such $dM^k_t=dN^k_t+{\tr} \, (\Ga_{N,t} L^*_kL_k)$
 and $N^k_t$ are counting processes with the intensities ${\tr} \, (\Ga_{N,t} L^*_kL_k)$.
The quantum filtering equation \eqref{eq2thquantumfileqcountpure} in terms of the pure states takes the form
\[
d \Psi_{N,t}=-iH(N) \Psi_{N,t}
    -\frac12 \sum_j((L_j^*-1)L_j-(L_j-\|L_j\Psi_{N,t}\|^2))\Psi_{N,t} \, dt
\]
\begin{equation}
\label{eqmainNpartBelnonlcount}
+\sum_j \left(\frac{L_j\Psi_{N,t}}{\|L_j\Psi_{N,t}\|}-\Psi_{N,t}\right)  dM_t^j.
\end{equation}

 The corresponding analog of the limiting equation \eqref{eqmainnonlinpartBel1dens} is the equation
\[
d \ga_{j,t}
=\left(-i [H_j,\ga_{j,t}]-i [A^{\bar\eta_t},\ga_{j,t}]
+(L_j\ga_{j,t} L_j^* -\frac12 L^*_jL_j \ga_{j,t} -\frac12 \ga_{j,t} L^*_jL_j)\right)\, dt
\]
\begin{equation}
\label{eqmainnonlinpartBel1denscount}
+ \left(\frac{L_j \ga_{j,t} L^*_j}{{\tr} \, (\ga_{j,t}L^*_jL_j)}- \ga_{j,t}\right) \, dM^j_t,
\end{equation}
where $M^j_t$ are martingales such $dM^k_t=dN^k_t+{\tr} \, (\ga_{j,t} L^*_jL_j)$
 and $N^k_t$ are counting processes with the intensities ${\tr} \, (\ga_{j,t} L^*_jL_j)$.

The corresponding equation for pure states is
\[
d \psi_{j,t}
=-i (H_j  +A^{\bar\eta_t})\psi_{j,t} \, dt
\]
\begin{equation}
\label{eqmainnonlinpartBel1count}
-\frac12 ((L_j^*-1)L_j-(L_j-\|L_j\psi_{j,t}\|^2))\psi_{j,t} \, dt
+\left(\frac{L_j\psi_{j,t}}{\|L_j\psi_{j,t}\|}-\psi_{j,t}\right)  dM_t^j.
\end{equation}

One sees directly the principle problem with this situation, which does not occur in the diffusive case. In
the latter both the approximating and the limiting equations are written  in terms of the same Brownian motions.
In the present case the driving noises are different (intensities of jumps depends on different objects).
To avoid this problem,  we shall deal only with a special conservative case when the operator $L$ is unitary:
 $L^*=L$. In this case equation   \eqref{eqmainNpartBeldensnonlcount} becomes linear, the intensities
 of all jump process $N_t^j$ become identical and equal to one (and thus not depending on a state of
 the process), so that the noises $M^j_t$ in
 \eqref{eqmainNpartBeldensnonlcount} and  \eqref{eqmainnonlinpartBel1denscount} can be identified.

\begin{remark}
The natural idea for dealing with the general case would be by organising certain coupling
between the counting processes $N^k_t$ from   equations  \eqref{eqmainNpartBeldensnonlcount}
and  \eqref{eqmainnonlinpartBel1denscount}. The natural coupling is the  marching coupling,
where the intensity of the coupled part is given by the minimum of the intensities of individual
noises. The author did not manage to prove convergence under such a coupling. The conjecture is that
it really does not hold. The reason is that one can prove convergence under the coupling that
makes the intensity of the coupled process to be equal to the maximum of the intensities of individual
noises. Of course such coupling destroys the individual dynamics.
\end{remark}

Thus let us look at the dynamics of $\al_{N,j}$,
\[
d \al_{N,j}=-{\tr} (d\Ga_{N,t} \ga_j(t))-{\tr} (\Ga_{N,t} \, d\ga_j(t))-{\tr} (d\Ga_{N,t} \, d\ga_j(t)),
\]
assuming $L=L^*$ and that the noises $M^j_t$ in
 \eqref{eqmainNpartBeldensnonlcount} and  \eqref{eqmainnonlinpartBel1denscount} are identified.

The part containing $H,A$ is the same as for diffusion. So we are interested only in the part containing $L_j$.

Recall the Ito multiplication rule for counting processes $dN_t^j dN_t^j=dN_t^j$.
Under assumptions of unitarity of $L$ it follows that
\[
dM_t^j dM_t^j=dN_t^j=dM_t^j+ {\tr} (\Ga_N L_j^*L_j) dt=dM_t^j+dt.
\]

The part of the stochastic differential (at $dM^j_t$) in the expression for $d \al_{N,j}$ is of no interest
for us, as we are looking for the expectation of $d \al_{N,j}$, which is not affected by these martingale terms.

It turns out that, under the unitarity assumption, the part at $dt$ depending on $L_j$ vanishes.
In fact (omitting index $t$ for brevity), taking into account that $\ga_j$ and $\Ga_N$
have unit traces, the coefficient at $dt$ depending on $L_j$ writes down as
\[
{\tr} \left[\frac12 \Ga_N \ga_j +\frac12 \Ga_N \ga_j-L_j \Ga_N L_j^*\ga_j
+\frac12 \Ga_N \ga_j +\frac12 \Ga_N\ga_j -\Ga_NL_j \ga_j L_j^*\right]
\]
\[
- {\tr} \left[(L_j \ga_j L^*_j- \ga_j)(L_j \Ga_N L^*_j- \Ga_N)\right]
\]
\[
={\tr} [2\ga_j \Ga_N  -\ga_j L_j\Ga_N L_j^*-\ga_j L^*_j \Ga_N L_j
-L_j\ga_j L_j^*L_j \Ga_N L_j^*+\ga_j L_j \Ga_N L_j^*+\ga_j L_j^*\Ga_N L_j-\ga_j \Ga_N]=0.
\]

Thus turning to the expectations of  $\al_{N,j}$ we have the same situation as in Theorem
\ref{thmynonlinSch}, but in the simpler version of the absence of $L$ in all estimates.
Consequently the following result holds.

\begin{theorem}
\label{thmynonlinSchcount}
Let $H,\hat H$ be self-adjoint operators in $L^2(X)$, with $\hat H$ bounded, and $L$ a unitary operator.
Let $A$ be a symmetric self-adjoint integral operator $A$ in $L^2(X^2)$ with a Hilbert-Schmidt
kernel $A(x,y;x',y')$ satisfying
\eqref{eq1thmynonlinSch} and \eqref{eq1athmynonlinSch}.
Let the function $u(t,\ga)$ with values in a bounded interval $[-U,U]$
be Lipschitz in the sense of
\eqref{eq1thmynonlinSchco}.

 Let $\psi_{j,t}$ be solutions to equations
\eqref{eqmainnonlinpartBel1count} with i.i.d. initial conditions
$\psi_{j,0}$, $\|\psi_{j,0}\|=1$.
Let $\Psi_{N,t}$ be a solution to the $N$-particle equation \eqref{eqmainNpartBelnonlcount}
with $H(N)$ of type \eqref{eqHambinaryinter1}, with some initial condition
$\Psi_{N,0}$, $\|\Psi_{N,0}\|_2=1$ such that
\[
\al_N(0)= \al_{N,j}(0)=1-\E \, {\tr} (\ga_{j,0} \Ga_{N,0})=1-\E \, {\tr} (\ga_{j,0} \Ga^{(j)}_{N,0})
\]
are equal for all $j$.
Then
\[
\E \al_N(t) \le \exp\{ (7\|A\|_{HS}+6 \ka \|\hat H\|)t\} \E \al_{N}(0)
\]
\[
+(\exp\{(7\|A\|_{HS}+6 \ka \|\hat H\|)t\}-1)\frac{1}{\sqrt N}.
\]

\end{theorem}

\section{Appendix A: a technical estimate}

\begin{lemma}
 \label{lemmaonspectraces}
 Let $\ga$ be a one-dim projector in a Hilbert space, $\Ga$ a density matrix (positive operator with unit trace)
 and $L$ a bounded operator in this Hilbert space.  Then
 \begin{equation}
\label{eq1lemmaonspectraces}
|-4\,{\tr}\, (L\ga L\Ga )+2 \,{\tr}\,(\Ga(L \ga+\ga L))\,{\tr}\,(\Ga L+\ga L)
-4 \,{\tr}\,(\Ga \ga)\, {\tr} (\Ga L)\, {\tr} (\ga L)|
\le 20 \|L\|^2 {\tr} ((1-\ga) \Ga)
\end{equation}
for a self-adjoint $L$, and
 \[
|-{\tr} (\ga L \Ga L^*+\ga L^* \Ga L +\ga L^*\Ga L^*+\ga L \Ga L)
\]
\[
+ {\tr} (\ga \Ga L^*+\ga L \Ga) \, {\tr} (\ga(L^*+L))
+ {\tr} (\ga \Ga L+ \ga L^* \Ga) \, {\tr} (\Ga_N(L^*+L))
\]
 \begin{equation}
\label{eq2lemmaonspectraces}
-{\tr} (\Ga \ga)\, {\tr} (\Ga(L^*+L))\, {\tr} (\ga(L^*+L))|
\le 28 \|L\|^2 \, {\tr}\, ((1-\ga) \Ga)
\end{equation}
 for a general $L$.
 \end{lemma}

 \begin{proof}
 By the approximation argument it is sufficient to prove the Lemma for a finite-dimensional Hilbert space $\C^n$.
 Let $\al= {\tr} ((1-\ga) \Ga)$.
 Let us choose an orthonormal basis, where $\ga$ is the projection on the first basis vector.

 By positivity of $\Ga$ it follows that
\begin{equation}
\label{eq2alemmaonspectraces}
|\Ga_{jk}|\le \al \, \text{for} \, j,k\neq 1, \,  \text{and}
\, \max(|\Ga_{j1}|,|\Ga_{1j}|)\le \sqrt \al \, \text{for} \, j\neq 1.
\end{equation}

Let $L$ be a self-adjoint matrix.
 Then the expression under the module sign on the l.h.s. of \eqref{eq1lemmaonspectraces} writes down as
 \[
  -4 (L\Ga L)_{11} +2 [(L \Ga)_{11}+(\Ga L)_{11}]({\tr}\,(\Ga L)+L_{11})-4 \Ga_{11} L_{11}\, {\tr}\,(\Ga L)
\]
\[
=-4\sum_{j,k} L_{1j}\Ga_{jk}L_{k1}+2[2L_{11}\Ga_{11}+\sum_{j\neq 1} (\Ga_{1j}L_{j1}+L_{1j}\Ga_{j1})]({\tr}\,(\Ga L)+L_{11})
-4 \Ga_{11} L_{11}\, {\tr}\,(\Ga L)
\]
\[
=-4 L_{11} \sum_{j\neq 1} (L_{1j}\Ga_{j1}+L_{j1}\Ga_{1j})-4\sum_{j\neq 1,k\neq 1} L_{1j}\Ga_{jk}L_{k1}
+2\sum_{j\neq 1} (\Ga_{1j}L_{j1}+L_{1j}\Ga_{j1})({\tr}\,(\Ga L)+L_{11})
\]
\[
=2\sum_{j\neq 1} (\Ga_{1j}L_{j1}+L_{1j}\Ga_{j1})L_{11}(\Ga_{11}-1)
-4\sum_{j\neq 1,k\neq 1} L_{1j}\Ga_{jk}L_{k1}
+2\sum_{j\neq 1} (\Ga_{1j}L_{j1}+L_{1j}\Ga_{j1})({\tr}\,(\Ga L)-L_{11}\Ga_{11})
\]
\[
=-2\sum_{j\neq 1} (\Ga_{1j}L_{j1}+L_{1j}\Ga_{j1})L_{11}\al
-4\sum_{j\neq 1,k\neq 1} L_{1j}\Ga_{jk}L_{k1}
\]
\[
+2\left(\sum_{j\neq 1} (\Ga_{1j}L_{j1}+L_{1j}\Ga_{j1})\right)^2
+2\sum_{j\neq 1} (\Ga_{1j}L_{j1}+L_{1j}\Ga_{j1})
\sum_{j\neq 1,k\neq 1} L_{kj}\Ga_{jk}.
\]
Here all terms are of order $\al$, because of \eqref{eq2alemmaonspectraces}.

More precisely,
\[
|\sum_{j\neq 1,k\neq 1} L_{kj}\Ga_{jk}|=|{\tr}[(1-\ga) L(1-\ga)\Ga(1-\ga)]|
\le \|L\| {\tr}[(1-\ga)\Ga]\le \|L\| \al.
\]
Moreover,
\[
\sum_{j\neq 1}|\Ga_{j1}|^2=\sum_{j\neq 1}|\Ga_{1j}|^2
\le \Ga_{11} \sum_{j\neq 1} \Ga_{jj}\le \Ga_{11}\al \le \al.
\]
Hence,
\[
|\sum_{j\neq 1} (\Ga_{1j}L_{j1})|^2
\le \sum_{j\neq 1} |\Ga_{1j}|^2 \sum_{j\neq 1} |L_{j1}|^2\le \|L\|^2 \al,
\]
\[
|\sum_{j\neq 1} (\Ga_{j1}L_{1j})|^2\le \|L^T\|^2 \al= \|L\|^2 \al,
\]
and thus
\begin{equation}
\label{eq3lemmaonspectraces}
|\sum_{j\neq 1} (\Ga_{1j}L_{j1}\pm L_{1j}\Ga_{j1})|\le 2\|L\| \sqrt \al.
\end{equation}
Finally,
\[
|\sum_{j\neq 1,k\neq 1} L_{1j}\Ga_{jk}L_{k1}|^2 \le \sum_{j,k} |L_{1j}|^2 |L_{k1}|^2
\sum_{j\neq 1,k\neq 1} |\Ga_{jk}|^2 \le \|L\|^4 (\sum_{j\neq 1} |\Ga_{jj}|)^2
\le \|L\|^4 \al^2,
\]
where the estimate
\[
|\Ga_{jk}|^2 \le \Ga_{jj}\Ga_{kk}
\]
for all $j,k$ was used (arising from the positivity of $\Ga$).

Putting the estimates together we get \eqref{eq1lemmaonspectraces}.

For a general $L$ we can write $L=L^s+L^a$, where $L^s=(L+L^*)/2$ is self-adjoint and
$L^a=(L-L^*)/2$ is anti-Hermitian. Plugging this into the l.h.s.
of \eqref{eq2lemmaonspectraces} leads to several cancelations, so that
 the expression under the module sign in the l.h.s. becomes equal to
\[
-4\,{\tr}\, (L^s\ga L^s\Ga )+2 \,{\tr}\,(\Ga(L^s \ga+\ga L^s))\,{\tr}\,(\Ga L^s+\ga L^s)
-4 \,{\tr}\,(\Ga \ga)\, {\tr} (\Ga L^s)\, {\tr} (\ga L^s)
\]
\[
+2{\tr} (\ga [\Ga, L^a])({\tr} (\Ga L_s)-{\tr} (\ga L^s)).
\]

Everything apart from the last term is already estimated  by \eqref{eq1lemmaonspectraces}.

By \eqref{eq3lemmaonspectraces} (that is valid for arbitrary $L$),
\[
|{\tr} (\ga [\Ga, L^a])|=|[\Ga,L^a]_{11}|=|\sum_{j\neq 1} (\Ga_{1j}L^a_{j1}-L^a_{1j}\Ga_{j1})|
\le 2\|L\| \sqrt \al,
\]
and
\[
|{\tr} (\Ga L_s)-{\tr} (\ga L^s)|
\]
\[
=|(\Ga_{11}-1)L^s_{11}+\sum_{j\neq 1} (\Ga_{1j}L^s_{j1}-L^s_{1j}\Ga_{j1})
+\sum_{j,k\neq 1} (\Ga_{jk}L^s_{kj})|\le 4 \|L\| \sqrt \al,
\]
which implies \eqref{eq2lemmaonspectraces}.
 \end{proof}

\begin{remark} Our proof of Lemma \ref{lemmaonspectraces} is based on some remarkable
cancellation of terms in concrete calculations
via coordinate representations. The author does not see any intuitive reasons for its validity.
Neither is it clear whether it can be extended to arbitrary density matrixes $\ga$,
not just one-dimensional projectors.
\end{remark}

\section{Appendix B: McKean-Vlasov diffusions in Hilbert spaces}

The McKean-Vlasov nonlinear diffusions are well studied processes, due to a large variety of applications.
However there seem to be only few publications for infinite-dimensional case,
see \cite{Ahmed16} and references therein. Therefore for completeness we provide here some basic results
for a class of McKean-Vlasov diffusions in Hilbert spaces stressing explicit bounds and errors.
 Namely, we are interested in the Cauchy problems
\begin{equation}
\label{eqMcKeVlBan}
dX_t=AX_t \, dt +b_t(X_t, \E X_t^{\otimes K}) \, dt +(\si (X_t), dB_t), \quad X_0=Y,
\end{equation}
in a complex Hilbert space $\HC$, equipped with the scalar product $(.,.)$ and
the corresponding norm $\|.\|$, and (as an auxiliary tool) in more standard equations
\begin{equation}
\label{eqMcKeVlBanau}
dX_t=AX_t \, dt +b_t(X_t, \xi_t) \, dt +(\si (X_t), dB_t), \quad X_0=Y.
\end{equation}
Here $K=1$ or $K=2$ (which are the most important cases for applications),
$B_t$ is the standard $n$-dimensional Wiener process,
defined on some complete probability space $(\Om, \FC, \P)$,
$\si(Y)=(\si_1(Y), \cdots, \si_n(Y))$
with each $\si_j$ a continuous mapping $\HC\to \HC$, $b$ a continuous mapping $\R \times \HC\times \HC^{\otimes K} \to \HC$,
$A$ a generator of a strongly continuous operator semigroup $e^{At}$ of contractions in $\HC$, $\xi_t$ a given
(deterministic) continuous curve in $\HC^{\otimes K}$. With some abuse of notations we shall denote by $(.,.)$
and $\|.\|$ also the scalar product and the norm in the tensor product Hilbert space $\HC^{\otimes K}$.
$\E$ in \eqref{eqMcKeVlBan} means the expectation with respect to the Wiener process.
A solution process to \eqref{eqMcKeVlBan} is called a nonlinear diffusion
of the McKean-Vlasov type, because of the dependence of the coefficient on such expectation.

To obtain effective bounds for growth and continuous dependence on parameters of equations
\eqref{eqMcKeVlBanau} we shall follow the strategy systematic developed in  \cite{Kolbook19}
for deterministic Banach space-valued equations. Namely,
we shall use the generalized fixed point principle of Weissinger type in the following form
(see e.g. \cite{Kolbook19}, Prop. 9.1 and 9.3 for simple proofs). If $\Phi$
is a mapping from a complete metric space $(M, \rho)$ to itself such that
\[
\rho(\Phi^k(x),\Phi^k(y)) \le \al_k \rho(x,y)
\]
with $a=1+\sum_j \al_j<\infty$, then $\Phi$ has a unique fixed point $x^*$,
and $\rho(x,x^*)\le a \rho(x, \Phi(x))$ for any $x$. Moreover, for any two mappings $\Phi_1, \Phi_2$
satisfying these conditions and such that $\rho(\Phi_1(x),\Phi_2(x)) \le \ep$ for all $x$,
the corresponding fixed points enjoy the estimate
\[
\rho (x_1^*, x_2^*)\le \ep a.
\]

Let us start with \eqref{eqMcKeVlBanau}. We shall work with the so-called mild form of this equation:
\begin{equation}
\label{eqMcKeVlBanaumi}
X_t=Y+\int_0^t e^{A(t-s)} [b_s(X_s, \xi_s) \, ds +(\si (X_s), dB_s)].
\end{equation}

For a Hilbert space $\BC$, which is either $\HC$ or $\HC^{\otimes 2}$, let $C_{ad}([0,T],\BC)$ denote
 the Banach space of adapted continuous
$\BC$-valued processes, equipped with the norm
\[
\|X_.\|_{ad,T}=\sup_{t\in [0,T]} \sqrt{\E \|X_t\|^2}.
\]
Its subspace of deterministic curves will be denoted   $C([0,T],\BC)$.
For elements $\xi_.$ of this subspace the norm turns to the standard sup-norm
\[
\|\xi_.\|_{ad,T}=\sup_{t\in [0,T]} \|\xi_t\|.
\]

By $C_{Y,ad}([0,T],\BC)$ and  $C_Y([0,T],\BC)$ we shall denote the
subsets of these spaces consisting of curves with $X(0)=Y$.

In our estimates we shall encounter the so-called Le Roy function of index $1/2$:
\begin{equation}
\label{eqLeRoy}
R(z)=\sum_{k=0}^{\infty} \frac{z^k}{\sqrt{k!}},
\end{equation}
which for stochastic equations play the same role as the exponential and the Mittag-Leffler function
for deterministic equations.

\begin{prop}
\label{propSDEBanach}
 Let
\begin{equation}
\label{eq1propSDEBanach}
\|b_t(Z_1, \xi_1)-b_t(Z_2,\xi_2)\|\le \ka_1 \|Z_1-Z_2\|+ \ka_2 \|\xi_1-\xi_2\|, \quad
\|\si(Z_1)-\si (Z_2)\|\le \ka_3 \|Z_1-Z_2\|.
\end{equation}
Then for any $T>0$, $Y\in \HC$, $\xi=\xi_. \in C([0,T],\HC^{\otimes K})$,
equation \eqref{eqMcKeVlBanaumi} has the unique global solution $X_. \in C_{ad}([0,T],\HC)$,
and it satisfies the estimate
\begin{equation}
\label{eq2propSDEBanach}
\|X_.-Y\|^2_{ad,T} \le 2t M^2(t)[\int_0^t |b_s(Y,\xi_s)|^2 \, ds+\si^2(Y)],
\end{equation}
where
\begin{equation}
\label{eq5propSDEBanach}
M(t)=R(\sqrt{2(\ka_3^2 +\ka_1^2)}\max(\sqrt t,t)).
\end{equation}
Moreover, for two initial conditions $Y_1,Y_2$
and two curves $\xi^1$, $\xi^2$ the corresponding solutions satisfy the estimate
\begin{equation}
\label{eq3propSDEBanach}
\|X^1_.-X^2_.\|^2_{ad,T} \le 2M^2(t)\left(\|Y^1-Y^2\|^2+t^2\ka_2^2 \|\xi^1_.-\xi^2_.\|^2_{ad,T} \right).
\end{equation}
\end{prop}

\begin{proof}
A solution to  \eqref{eqMcKeVlBanaumi} is a fixed point of the mapping
\begin{equation}
\label{eq4propSDEBanach}
[\Phi_{Y,\xi}(X_.)](t)=Y+\int e^{A(t-s)} [b_s(X_s, \xi_s) \, ds +(\si (X_s), dW_s)]
\end{equation}
in $C_{Y,ad}([0,T])(\HC)$.
For two curves $X^1_.,X^2_.\in C_{Y,ad}([0,T])(\HC^{\otimes K})$ we have
by the contraction property of the semigroup $e^{At}$ and Ito's isometry that
\[
\E \|[\Phi_{Y,\xi}(X^1_.)](t)-[\Phi_{Y,\xi}(X^2_.)](t)\|^2
\]
\[
\le 2 \E \|\int_0^t e^{A(t-s)} (b_s(X^1_s, \xi_s) - b_s(X^2_s, \xi_s))\, ds\|^2
+2 \E \|\int_0^t e^{A(t-s)} (\si(X^1_s) - \si(X^2_s))\, dW_s\|^2
\]
\[
\le 2 \E \left(\int_0^t \ka_1 \|X^1_s-X_s^2\|\, ds\right)^2
+2 \ka_3^2 \E \int_0^t \|X^1_s - X^2_s\|^2\, ds.
\]
Applying Cauchy-Schwartz to the first integral we get
\[
\left(\int_0^t \ka_1 \|X^1_s-X_s^2\|\, ds\right)^2
\le \ka_1^2 t \int_0^t \|X^1_s-X_s^2\|^2\, ds
\]
and thus
\[
\E \|[\Phi_{Y,\xi}(X^1_.)](t)-\Phi_{Y,\xi}(X^2_.)](t)\|^2
\le 2 (\ka_3^2 +\ka_1^2 t)\E \int_0^t \|X^1_s - X^2_s\|^2\, ds.
\]
This estimate is easy to iterate. Namely, for $t\le 1$ we get for the $k$th power of $\Phi$ the estimate
\[
\E \|[\Phi^k_{Y,\xi}(X^1_.)](t)-[\Phi^k_{Y,\xi}(X^2_.)](t)\|^2
\le 2^k (\ka_3^2 +\ka_1^2)^k \frac{t^k}{k!}\|X^1_. - X^2_.\|^2_{ad,T},
\]
so that
\[
\|[\Phi^k_{Y,\xi}(X^1_.)-\Phi^k_{Y,\xi}(X^2_.)]\|_{ad,T}
\le \frac{1}{\sqrt {k!}}[2 (\ka_3^2 +\ka_1^2) t]^{k/2}\|X^1_. - X^2_.\|_{ad,T}.
\]
And for $t\ge 1$ we get for the $k$th power of $\Phi$ the estimate
\[
\E \|[\Phi^k_{Y,\xi}(X^1_.)](t)-[\Phi^k_{Y,\xi}(X^2_.)](t)\|^2
\le 2^k (\ka_3^2 +\ka_1^2)^k \frac{t^{2k}}{(2k-1)!!}\|X^1_. - X^2_.\|^2_{ad,T},
\]
so that
\[
\|[\Phi^k_{Y,\xi}(X^1_.)-\Phi^k_{Y,\xi}(X^2_.)]\|_{ad,T}
\le \frac{1}{\sqrt {k!}}[2 (\ka_3^2 +\ka_1^2) t^2]^{k/2}\|X^1_. - X^2_.\|_{ad,T}.
\]
Therefore the conditions of the generalized fixed point theorem above is satisfied
for $a=M(t)$ implying all statements of the theorem.
\end{proof}

A solution to the mild form
\begin{equation}
\label{eqMcKeVlBanmi}
X_t=Y+\int_0^t e^{A(t-s)} [b_s(X_s, \E X_s^{\otimes K}) \, ds +(\si (X_s), dB_s)]
\end{equation}
of equation \eqref{eqMcKeVlBan} is a fixed point of the mapping
$\Ga:  C_Y([0,T], \HC^{\otimes K})\to  C_Y([0,T], \HC^{\otimes K})$
that maps $\xi_.$ to $\E X_.^{\otimes K}$, where $X_.$ is the solution to \eqref{eqMcKeVlBanaumi}.
By \eqref{eq3propSDEBanach}, for two curves $\xi^1$ and $\xi_2$ we get for the corresponding
solutions $X_.^1$ and $X_.^2$ the estimates
\[
\| \E (X_t^1)-\E (X_t^2) \|
\le  2 M(t)t\ka_2 \|\xi^1_.-\xi^2_.\|_{ad,t}.
\]
for $K=1$ and
\[
\| \E (X_t^1)^{\otimes 2}-\E (X_t^2)^{\otimes 2} \|
\le \E \|(X_t^1-X_t^2)\otimes (X^1_t)\|+\E \|(X^2_t)\otimes (X_t^1-X_t^2)\|
\]
\[
\le  2 M(t)t\ka_2 \|\xi^1_.-\xi^2_.\|_{ad,t} \max(\|X^1_.\|_{ad,t}, \|X^2_.\|_{ad,t})
\]
for $K=2$.

Hence for small times $t\le t_0$ the mapping $\Ga$ is a contraction and hence has a unique fixed point.
For $K=1$ the time $t_0$ does not depend on $Y$ and hence one can build a unique global solution by iterations.
Consequently we proved the following statement.

\begin{prop}
\label{propSDEBanach2}
 Let the assumptions of Proposition \ref{propSDEBanach} hold. If $K=1$,
 equation \eqref{eqMcKeVlBanmi} has a unique global solution
 for any initial $Y$, and
 \[
 \| \E X_. -Y\|_{ad,T} \le C(T)
\]
with a constant $C(T)$ depending on $\ka_1, \ka_2, \ka_3, \|Y\|$.
If $K=2$ equation \eqref{eqMcKeVlBanmi} has a unique local solution
for times of order $\|Y\|^{-1}$.
\end{prop}

Finally let us point out a situation where solutions to mild equations solve also the initial SDEs.

 \begin{prop}
\label{propSDEBanach3}
Let $D$ be an invariant core for the semigroup $e^{At}$ , which is itself a
Banach space with respect to some norm $\|.\|_D$. Let $b$ and $\si$ be continuous
mappings $\R\times \HC\times \HC^{\otimes K} \to D$ and $\HC\to D$. Then for any $Y\in D$ the solutions to mild equations
\eqref{eqMcKeVlBanmi} and \eqref{eqMcKeVlBanaumi} solve the corresponding SDEs.
\end{prop}

\begin{proof} This follows from the direct application of Ito's rule. The justification of the differentiation
required is the consequence of the assumptions made.
\end{proof}

\end{document}